\documentclass[fullpage,10pt]{article}
\usepackage{authblk}

\usepackage{amsmath,amsthm,amssymb}

\newcommand{\N}{\mathbb{N}}

\newcommand{\R}{\mathbb{R}}

\newcommand{\Sym}{\mathcal{S}}
\newcommand{\X}{\mathcal{X}}

\newcommand{\diag}{\operatorname{Diag}}
\newcommand{\im}{\operatorname{Im}}

\newcommand{\old}[1]{{}}

\newcommand{\aff}{\operatorname{aff}}

\newtheorem{nn}{}
\newtheorem{lemma}[nn]{Lemma}
\newtheorem{theorem}[nn]{Theorem}

\newtheorem{prop}[nn]{Proposition}
\newtheorem{remark}[nn]{Remark}
\newtheorem{definition}[nn]{Definition}

\newtheorem{fact}{Fact}
\newtheorem{question}{Question}

\title{Computing approximate PSD factorizations}

\author{Amitabh Basu \thanks{Supported in part by the NSF grant CMMI1452820.}}
\affil{Deptartment of Applied Mathematics and Statistics, Johns Hopkins University}
\author{Michael Dinitz \thanks{Supported in part by NSF awards CCF-1464239 and CCF-1535887.}}
\author{Xin Li}
\affil{Deptartment of Computer Science, Johns Hopkins University}
\begin{document}

\maketitle

\begin{abstract} We give an algorithm for computing approximate PSD factorizations of nonnegative matrices.  The running time of the algorithm is polynomial in the dimensions of the input matrix, but exponential in the PSD rank and the approximation error. The main ingredient is an {\em exact} factorization algorithm when the rows and columns of the factors are constrained to lie in a general polyhedron. This strictly generalizes nonnegative matrix factorizations which can be captured by letting this polyhedron to be the nonnegative orthant.
\end{abstract}

\section{Introduction}

Matrix factorization is a fundamental operation that has importance for diverse areas of mathematics and engineering such as machine learning, communication complexity, polyhedral combinatorics, statistical inference, and probability theory, to name a few. The problem can be stated quite simply as follows: \begin{quote}
  Given two sequences of sets $\mathcal{K} = \{K_d\}_{d\in \N}$ and $\mathcal{K}' = \{K'_d\}_{d\in \N}$ where $K_d, K'_d$ are subsets of $\R^d$ for all $d\in \N$, and a matrix $M \in \R^{n\times m}$, find a factorization $M = U V$ where $U \in \R^{n\times d}$ and $V \in \R^{d\times m}$, and each row of $U$ is in $K_d$ and each column of $V$ is in $K'_d$. \end{quote}

Such a factorization is called a {\em $\mathcal{K},\mathcal{K'}$ factorization}. The smallest $d\in \N$ such that such a factorization exists is called the {\em $\mathcal{K},\mathcal{K'}$ rank}. Most of the literature on this problem focuses on the case when the matrix $M$ is nonnegative. In this context, when $K_d = K'_d = \R^d_+$, the factorization is called {\em nonnegative factorization}, and the corresponding rank is called {\em nonnegative rank}. When $K_d = K'_d$ are the cone of $d\times d$ PSD matrices, the factorization is known as a {\em PSD factorization} and the corresponding rank is called the {\em PSD rank}. These notions will be the object of study in this paper. A more general notion is that of {\em cone factorizations}, where $\mathcal{K}$ is a family of cones and $\mathcal{K}'$ is the family of corresponding dual cones; see~\cite{gouveia2013lifts}.

One of the most elegant applications of such factorizations arises in combinatorial optimization. A very common technique in approaching combinatorial optimization problems is to formulate the problem as a linear programming problem. However, a naive formulation of a problem may result in a polytope (the feasible region of the LP) with a large number of facets (exponentially many in the size of the problem), making it intractable to actually solve.  One way around this is to try to  express the polytope as the projection of a higher dimensional convex set.  In particular, suppose that it can be expressed as the projection of either a higher dimensional polytope (LP), the feasible region of an SDP, or the feasible region of a more general convex optimization problem. Furthermore suppose that the number of ``extra'' dimensions is polynomial in the size of the original problem, and the description of the higher-dimensional convex optimization problem is also polynomial in the size of the original problem (i.e.~there are not an exponential number of facets). Then we can efficiently solve the higher-dimensional problem, which means we can efficiently solve the original LP, even if its size makes solving it directly intractable. 

It turns out that the smallest size of such a reformulation is a direct function of the nonnegative rank (for LP reformulations), the PSD rank (for SDP reformulations), or more general cone factorization ranks of the so-called {\em slack matrix} of the original LP formulation. The actual factorization can be used to explicitly find the smallest reformulation. This line of research started with a seminal paper by Yannakakis~\cite{yannakakis1988expressing}, and has recently seen a flurry of research activity -- see the surveys~\cite{kaibel2011extended,conforti2010extended} and~\cite{fiorini2012linear,fiorini2013combinatorial,rothvoss2013some,briet2014existence,lee2014power,rothvoss2014matching,lee2014lower,braun2015matching} for some of the most recent breakthroughs.
 
In machine learning applications the actual factorization is perhaps more important than the value of the rank, as this factorization is key to certain text mining, clustering, imaging and bioinformatics applications. A key algorithmic question is computing such a factorization. Unfortunately, this question is computationally challenging -- even computing the nonnegative rank was proved to be NP-hard by Vavasis~\cite{vavasis2009complexity}. 

A recent algorithmic breakthrough was achieved by Arora et al~\cite{arora2012computing}, where they showed that computing nonnegative factorizations can be done in polynomial time (in the dimensions of the input matrix) for the family of matrices with fixed (constant) nonnegative rank. The running time of their algorithm was doubly exponential in the nonnegative rank, and this was later improved to a singly exponential algorithm by Moitra~\cite{moitra2013almost}, which he showed to be nearly optimal under the Exponential Time Hypothesis. The analogous question for PSD factorizations is largely open (the question is also posed in the survey~\cite{psd-survey}):

\begin{question}\label{quest:1}
Let $r \in \N$ be a constant. Does there exist an algorithm which, given any $n\times m$ nonnegative matrix $M$ with PSD rank $r$, computes a PSD factorization of rank $r$ in time polynomial in $n,m$?
\end{question}

 Our main result is a polynomial time algorithm to compute {\em approximate} factorizations of matrices with fixed PSD rank. We consider the space $\Sym^r$ of $r\times r$ symmetric matrices, and the cone of $r\times r$ PSD matrices in this space, denoted by $\Sym^r_+$. Given any matrix $M \in \R^{n\times m}$, we use the notation $\| M \|_\infty := \max_{i,j} |M_{ij}|$.

More precisely, we prove:

\begin{theorem}\label{thm:factor}
Let $r \in \N$ be fixed. Then there exists an algorithm which, given any $\epsilon >0$ and any $n\times m$ nonnegative matrix $M$ with PSD rank $r$, computes a factorization $M=U V$ such that each row of $U$ and each column of $V$ are in $\Sym^r_+$ such that $$\| M - U V\|_\infty \leq \epsilon\| M \|_\infty$$ and has runtime polynomial in the dimensions of $M$. %({\tt Amitabh -- Should check if we can say strongly polynomial time}).
\end{theorem}

Approximate PSD factorizations can be useful for reformulation questions in combinatorial optimization, where one seeks approximations of the original polyhedron using SDPs, as opposed to an exact reformulation -- see~\cite{gouveia2015approximate} for results along this direction.  In particular, approximate factorizations of the slack matrix of a polytope can sometimes be used to compute ``inner" and ``outer" approximations of the polytope, each of which can then be optimized over in order to give an approximation to the true optimal solution of the polytope. However, in~\cite{gouveia2015approximate}, these approximations are guaranteed only when the corresponding matrix factorization error is calculated in certain induced matrix norms (in particular the $\|\cdot\|_{1,2}$ and $\|\cdot\|_{1,\infty}$ norms). Consequently, it is unclear if the approximate factorizations generated by Theorem~\ref{thm:factor} give similar results, primarily since our notion of an approximate factorization involves the $\| \cdot \|_{\infty}$-norm rather than the appropriate induced matrix norms.  Thus while our results do not directly imply any new approximation algorithms, they do provide ideas on how to go beyond nonnegative factorizations to PSD or more general conic factorizations, if one admits approximate factorizations as opposed to exact ones.

%shed light on the structure of PSD factorizations, and on the differences between nonnegative factorizations and PSD or more general conic factorizations.  \mdnote{Check that all of this is true.}

%Moreover, approximate PSD factorizations lead to protocols with provable guarantees in quantum communication complexity -- {\color{red} Is this even true? Don't know a good citation for this}.

\subsection{Technical overview.} Our algorithm for Theorem~\ref{thm:factor} is inspired by ideas behind the algorithm in Arora et al~\cite{arora2012computing}. However, there are some important differences. Arora et al's algorithm uses properties of the nonnegative orthant that do not hold for the cone of PSD matrices. To overcome this difficulty, we need to approximate the PSD cone by a polyhedral cone obtained by intersecting enough tangent halfspaces. We then generalize Arora et al's techniques to compute factorizations inside a general polyhedron, as opposed to just the nonnegative orthant. The nonnegative orthant is a very special polyhedron, and many of its special properties are utilized in the algorithm of Arora et al. We have to use interesting techniques from polyhedral theory (such as Fourier-Motzkin elimination) to extend these ideas to handle general polyhedra (see Theorem~\ref{thm:arora-gen}). Finally, to bound the errors in the approximate factorization, we use some technical results on rescaling PSD factorizations due to Briet et al~\cite{briet2014existence} (Theorem~\ref{thm:rescaling}).
\bigskip

\subsection{\bf Model of computation.} We will present our algorithm from Theorem~\ref{thm:factor} in the real arithmetic model of computation developed by Blum, Shub and Smale~\cite{blum1989theory}, thus ignoring questions of approximating irrational computations by rational arithmetic. This is just for the ease of exposition. In Section~\ref{sec:turing}, we show that Theorem~\ref{thm:factor} can be proved by designing an algorithm that operates in the more standard Turing machine model of computation.

\section{Preliminaries}

%For general matrices $M \in \R^{n\times m}$

For any normed space $(V, \lVert \cdot \rVert)$, we denote the distance between two subsets $X,Y\subseteq V$ by $dist(X,Y): = \inf\{\lVert x - y \lVert : x \in X, y \in Y\}$.
%We first recall some easy facts:
A closed subset $P$ of a normed space $V$ is called a {\em closed cone} if it is convex and $\lambda P \subseteq P$ for all $\lambda \geq 0$. A cone is called a {\em polyhedral cone} if it is the intersection of finitely many halfspaces. For any closed cone $P$ in an inner product space $(V, \langle \cdot, \cdot \rangle)$ (closed with respect to the norm obtained from the inner product), the dual cone will be denoted by $$P^* = \{v \in V: \langle v, y\rangle \geq 0 \;\; \forall y \in P\}.$$ %The dual cone is easily seen to be a cone in $V$ also.
%\end{definition}

%\mdnote{$\cdot$ has now been used both for inner product and matrix multiplication.  Maybe use it only for one or the other? Might be confusing since we also talk about inner product of matrices.} 

We recall a standard fact about dual cones:
\begin{fact}\label{fact:sep-hyperplane}% in an inner product space 

Let $(V, \langle \cdot, \cdot \rangle)$ be an inner product space with $\lVert \cdot \rVert$ denoting the norm on $V$ induced by the inner product. %the separating hyperplane theorem \mdnote{What is this theorem?} implies the following:
%\begin{enumerate}
%\item $(P^*)^* = P$. In particular, if there exists a point $x\not\in P$, there is a vector $a \in P^*$ such that $a\cdot x < 0$. \mdnote{These aren't equivalent statements, right?  The first implies the second, but not the other way?}
%\item 
 For any closed cone $P\subseteq V$, if $x \in V$ such that $dist(x,P) = \delta$, then there exists a vector $a \in P^*$ with $\lVert a \rVert = 1$ such that the distance of $x$ from the hyperplane $\{y \in V: \langle a, y \rangle= 0\}$ is $\delta$, i.e., $\langle a, x \rangle= -\delta$. 
%\end{enumerate}
\end{fact}

On the space $\Sym^r$ of $r\times r$ symmetric matrices, we consider the inner product $\langle A, B \rangle = \sum_{i,j}A_{ij}B_{ij}$.

\begin{fact}\label{fact:PSD-self-dual}
The PSD cone $\Sym^r_+$ is self-dual, i.e.,  $(\Sym^r_+)^* =  \Sym^r_+$. %\mdnote{Is it standard what dot product of matrices is?  Particularly since previously we use $U \cdot V$ to denote normal matrix multiplication}
\end{fact}

%We choose an arbitrary norm $\lVert \cdot \rVert$ on the space $\Sym^r$ (e.g, one can choose the Frobenius norm or the spectral norm or any matrix norm).
\begin{definition}\label{def:cap} Let $C$ be a subset of a normed space $(V,\lVert \cdot \rVert$). For $\epsilon > 0$, $\X_\epsilon \subseteq C$ is called an $\epsilon$-covering for $C$ with respect to the norm $\lVert\cdot \rVert$ if for every $a \in C$, there exists $a' \in \X_\epsilon$ such that $\lVert a - a' \rVert < \epsilon$.
\end{definition}

\begin{definition}
For any closed cones $P_1 \subseteq P_2$ in a normed space $(V, \lVert \cdot \rVert)$, we say $P_2$ is an $\epsilon$-approximation of $P_1$ with respect to $\lVert\cdot \rVert$ for some $\epsilon > 0$, if for every $p_2 \in P_2$, there exists a point $p_1 \in P_1$ such that $\lVert p_2 - p_1 \rVert \leq \epsilon\lVert p_2\rVert$. 
\end{definition}

\begin{theorem}\label{thm:eps-approx}
Let $C = \{x \in \Sym^r_+: \lVert x \rVert_2 = 1\}$ be the spherical cap on the PSD cone. Let $\epsilon > 0$ and let $\X_\epsilon \subseteq C$ be any {\em finite} $\epsilon$-covering for $C$ with respect to some norm $\lVert \cdot \rVert$. Then the polyhedral cone $$P:= \{x \in \Sym^r : \langle a', x\rangle \geq 0 \;\; \forall a' \in \X_\epsilon\}$$
is an $\epsilon$-approximation for $\Sym^r_+$ with respect to $\lVert\cdot \rVert$.
\end{theorem}

\begin{proof}
It suffices to prove that for any $x \in \Sym^r$ such that $dist(x,\Sym^r_+) > \epsilon\lVert x\rVert$, then $x \not\in P$. By Fact~\ref{fact:sep-hyperplane} and Fact~\ref{fact:PSD-self-dual}, there exists $a \in C$ such that $\langle a, x\rangle < -\epsilon\lVert x\rVert$. By definition of $\epsilon$-covering, there exists $a' \in \X_\epsilon$ such that $\lVert a - a' \rVert < \epsilon$. By Cauchy-Schwartz, we have $|\langle a', x\rangle - \langle a, x\rangle| \leq \lVert a - a' \rVert \lVert x\rVert <\epsilon\lVert x\rVert$. Combined with $\langle a, x\rangle < -\epsilon\lVert x\rVert$, this implies $\langle a', x\rangle < 0$. Thus, by definition of $P$, $x \not\in P$. 
\end{proof}

\begin{remark}\label{rem:explicit-construct}
Since $C$ is a compact set, there always exists a finite $\epsilon$-covering of $C$ for any $\epsilon > 0$. Rabani and Shpilka~\cite{rabani2010explicit} give explicit constructions of small $\epsilon$-coverings of the sphere $\mathcal{S}^{d-1} = \{x\in \R^d: \|x\|_2=1\}$, which will prove useful for us.
\end{remark}

%\section{Facts from polyhedral theory}

The following fact from linear algebra is useful.

\begin{prop}\label{prop:lin-trans}
Any linear transformation $T : \R^d \to \R^m$ can be expressed as $T=A\circ \phi$ where $\phi:\R^d \to \ker(T)^\perp$ is the projection of $\R^d$ onto $\ker(T)^\perp$ and $A: \ker(T)^\perp\to \im(T)$ is an invertible linear transformation.\end{prop}

This leads to the following observation about linear transformation of polyhedra.

\begin{prop}\label{prop:small-desc} Let $P \subseteq \R^d$ be a polyhedron defined by $p$ inequalities. Let $T: \R^d \to \R^m$ be any linear transformation. Then $T(P)$ is a polyhedron defined by at most $O(p^{2^d})$ inequalities.
\end{prop}
\begin{proof}
Let us make a change of coordinates such that $\ker(T)^\perp = \R^{d'}$ with $d \geq d' \geq 0$ - this does not change the number of inequalities required to describe $P$ or $T(P)$. By Proposition~\ref{prop:lin-trans}, $T$ can be expressed as $A\circ \phi$ where $\phi$ is the projection from $\R^d \to \R^{d'}$, and $A$ is an invertible transformation from $\R^{d'} \to \im(T)$.  So we just need to analyze the effect of $\phi$ and $A$ on the number of inequalities.

 To analyze $\phi(P)$, we note that the Fourier-Motzkin elimination process~\cite{ziegler} implies that projecting out a single variable can be done by squaring the number of inequalities.  By repeatedly applying this, we get that $\phi(P)$ has at most $p^{2^{d-d'}}$ inequalities.  Since $A$ is an invertible linear transformation, $A(\phi(P))$ has the same number of inequalities as $\phi(P)$. The result follows. 
\end{proof}

We list one final linear algebraic observation. Let $\dim(W)$ denote the dimension of an affine subspace $W$, and let $\aff(X)$ denote the affine hull of the columns of a matrix $X$ (or just a finite set of vectors $X$). 

\begin{prop}\label{prop:invert}
Let $\{m_1, \ldots, m_t\} \subseteq \R^m$ and $\{b_1, \ldots, b_t\} \subseteq \R^d$ such that there exists a linear transformation $A:\R^d \to \R^m$ such that $m_i = A(b_i)$ for all $i = 1, \ldots, t$. Further suppose that $\dim(\aff(\{m_1, \ldots, m_t\})) = \dim(\aff(\{b_1, \ldots, b_t\})) = k$ and that $m_1, \ldots, m_{k+1}$ and $b_1, \ldots, b_{k+1}$ are maximal affinely independent subsets, respectively. 

Then, for every $i > k+1$, $m_i = \lambda_1 m_1 + \ldots + \lambda_{k+1}m_{k+1}$ implies that $b_i = \lambda_1 b_1 + \ldots + \lambda_{k+1}b_{k+1}$. %\mdnote{Should $v_i$'s be $b_i$'s?}
\end{prop}

The following result about projecting onto the PSD cone will be used~\cite{henrion2012projection}.

\begin{prop}\label{prop:project}
Let $C$ be an $r\times r$ symmetric matrix with spectral decomposition $U\Lambda U^T$, where $U$ is the matrix with the eigenvectors of $C$ as columns, and $\Lambda = \diag(\lambda_1, \ldots, \lambda_r)$ is the diagonal matrix with eigenvalues of $C$ on the diagonals. Then, the matrix $C^* = U\left(\diag(\max\{0,\lambda_1\}, \ldots, \max\{0,\lambda_r\})\right) U^T$ is the closest matrix in $\mathcal{S}^r_+$ to $C$ with respect to the $\|\cdot\|_2$ norm.
\end{prop}

We also use the following deep result from real algebraic geometry and quantifier elimination.

\begin{theorem}\label{thm:quantifier-elimination}\cite{basu1996combinatorial}
%Let $N \in \N$ be fixed. 
There is an algorithm that tests the feasibility of any system of $s$ polynomial equalities involving $N$ variables with $d$ as the maximum degree of any polynomial, that runs in time $(sd)^{O(N)}$.\end{theorem}

\section{Factorizations from a polyhedron}

Our main tool for proving Theorem~\ref{thm:factor} will be the following generalization of the algorithm of Arora et al.~\cite{arora2012computing}, who proved it for the special case of $P$ being the nonnegative cone. We generalize this to an arbitrary polyhedron $P$.

\begin{theorem}\label{thm:arora-gen}
Let $M$ be an $n\times m$ matrix with nonnegative entries, and let $P$ be some polyhedron in $\R^d$ described by $p$ inequalities. If there exists a factorization $M = U V$ such that each row of $U$ and each column of $V$ is in $P$, then one can compute such a factorization in time polynomial in $n$ and $m$ (assuming $d$ and $p$ to be constants).
\end{theorem}

In order to prove this theorem, we first need a few useful lemmas.  Let $X$ be a $p\times q$ matrix. For any subset $C \subseteq \{1, \ldots q\}$, let $X^C$ denote the matrix formed by the subset of columns indexed by $C$. Similarly, for any subset $R \subseteq \{1, \ldots p\}$, let $X_R$ denote the matrix formed by the rows indexed by $R$. %\mdnote{Reworded -- this is still accurate, right?}

\begin{lemma}\label{lem:partition}
Let $M$ be an $n\times m$ matrix with nonnegative entries. Let $P$ be some polyhedron in $\R^d$ described by $p$ inequalities. Suppose there exists a factorization $M = U V$ such that each row of $U$ and each column of $V$ is in $P$. Then there exists a partition $C_1 \uplus C_2 \uplus \ldots \uplus C_k = \{1, \ldots, m\}$, a partition $R_1 \uplus R_2 \uplus \ldots \uplus R_\ell = \{1, \ldots, n\}$, and matrices $\bar{U} \in \R^{n\times d}$, $\bar{V} \in \R^{d \times m}$ such that the following properties all hold:
\begin{enumerate}
\item $M = \bar{U}  \bar{V}$.
\item Each row of $\bar U$ and each column of $\bar V$ is in $P$.
\item $\dim(\aff(M^{C_j})) = \dim(\aff(\bar{V}^{C_j}))$ for all $j = 1, \ldots, k$.
\item $\dim(\aff((M_{R_i})^T)) = \dim(\aff((\bar{U}_{R_j})^T))$ for all $j = 1, \ldots, \ell$.
\item $k, \ell \leq p^d$.
\end{enumerate}

%Moreover, $k, \ell$ are both bounded above by the number of subsets of vertices of $P$, which is bounded by $O(2^{p^d})$.

\end{lemma}

\begin{proof}
%Xin's proof goes here.

We use an idea from Arora et al.~\cite{arora2012computing} to produce $\bar U, \bar V$ and the partitions with the stated properties. Starting from $U, V$, we will first construct $\bar V$, and then use this to construct $\bar U$.  Slightly more formally, we will first construct a partition $C_1 \uplus C_2 \uplus \ldots \uplus C_k = \{1, \ldots, m\}$ and $\bar V$ such that $M = U \bar V$, all columns of $\bar V$ are in $P$, and condition 3 in the statement is satisfied.  We will then keep $\bar V$ fixed and will construct a partition $R_1 \uplus R_2 \uplus \ldots \uplus R_\ell = \{1, \ldots, n\}$ such that each row of $\bar U$ is in $P$, and condition 4 from the statement is satisfied.  Condition 5 will then be straightforward.

For any $p \in P$, let $F_p$ be the face of $P$ of minimum dimension containing $p$. This induces a partial ordering $\succ$ on the points in $P$, where $p_1 \succ p_2$ if $F_{p_1} \supsetneq F_{p_2}$. 

For every column $v^j$ of $V$, consider the set $(v^j + \ker(U))\cap P$ and define $\bar{v}^j$ to be a minimal element in this set according to this partial order. Note that for any $p \in P$, if there exists $u \in \ker(U)\setminus\{0\}$ such that the line $p + \lambda u$, $\lambda \in \R$ lies in the affine hull of $F_p$, then one can choose $\lambda$ such that $p + \lambda u$ is in a strict face of $F_p$. Thus, by the minimal choice of $\bar{v}^j$, we have that $(\bar{v}^j + \ker(U)) \cap \aff(F_{\bar{v}^j}) = \bar{v}^j$ for every $j \in \{1, \ldots, m\}$. We set $\bar V$ to be the matrix with columns $\bar{v}^j$.  Note that $M = U  \bar V$ as desired, since $U \bar v^j = U(v^j + x^j) = U v^j$ for every $j \in \{1,\ldots, m\}$, where $x^j$ is some vector in $\ker(U)$.  

The partition $C_1 \uplus C_2 \uplus \ldots \uplus C_k$ of the columns of $\bar V$ is obtained by grouping the columns together based on the face of minimum dimension that they lie on. Thus, $k \leq p^d$ which is an upper bound on the number of faces of $P$. We now need to verify that $\dim(\aff(M^{C_j})) = \dim(\aff(\bar{V}^{C_j}))$ for all $j = 1, \ldots, k$. Fix some $j$ and let the columns of $M^{C_j}$ be $\{m_0, m_1, \ldots, m_h\}$ and let the columns of $\bar{V}^{C_j}$ be $v_0, v_1, \ldots, v_h$. 
Since $M^{C_j} = U\bar{V}^{C_j}$, we know that $\dim(\aff(M^{C_j})) \leq \dim(\aff(\bar{V}^{C_j}))$. If the inequality is strict, then there exists a $v \in \ker(U)\setminus\{0\}$ such that $v = \lambda_0v_0 + \lambda_1v_1 + \ldots, \lambda_hv_h$ and $\lambda_0 + \lambda_1 + \ldots + \lambda_h = 0$. But then, if $F$ is the face of minimum dimension containing $v_0, v_1, \ldots, v_h$, we find that $v_0 + \lambda v$, $\lambda \in \R$ lies in the affine hull of $F$.  This would contradict the construction of the columns of $\bar V$. Therefore, $\dim(\aff(M^{C_j})) = \dim(\aff(\bar{V}^{C_j}))$.

In a similar manner, we can change the rows of $U$ (keeping $\bar V$ fixed) to obtain $\bar U$ so that condition $1$ still holds and condition 2 is now satisfied.  We can also construct the partition $R_1 \uplus R_2 \uplus \ldots \uplus R_\ell = \{1, \ldots, n\}$ (in the same way as the column partition) so that property 4 is satisfied.  Finally, note that the parts in the partition are in correspondence with faces of $P$ (as was the case with the column partition), giving $\ell \leq p^d$. This completes the construction. 
\end{proof}

Let $X$ be a a set of points in $\R^d$. We say a set of polyhedra $P_1, \ldots, P_k$ is a {\em polyhedral covering of $X$} if $(P_1 \cap X) \cup \ldots \cup (P_k\cap X) = X$. We say a partition $ X_1 \uplus \ldots \uplus X_k = X$ is {\em induced by a polyhedral covering} if there exists a polyhedral covering $P_1, \ldots, P_k$ of $X$ and $X_1 = P_1\cap X$ such that $X_i = (P_i \cap X_i) \setminus (X_1 \cup \ldots \cup X_{i-1})$ for $i=2, \ldots, k$. A {\em $(k_1, k_2)$-polyhedral partition of $X$} is a partition induced by a polyhedral covering of $X$ with at most $k_1$ polyhedra and each polyhedron is described by at most $k_2$ inequalities.%A {\em hyperplane partition} of $X$ given by a hyperplane $a\cdot x = b$ is a partition induced by a polyhedral covering $P_1, P_2$ such that $P_1 = \{x \in \R^d: a\cdot x \leq b\}$ and $P_2 = \{x \in \R^d: a\cdot x \geq b\}$.

\begin{lemma}\label{lem:poly-partition}
Let $k_1, k_2$ be fixed natural numbers and let $X$ be a set of points in $\R^d$. The number of $(k_1,k_2)$-polyhedral partitions is at most $O((2^dm^{d})^{k_1k_2})$ and one can enumerate these partitions in time $O((2^dm^{d})^{k_1k_2})$, where $m = |X|$.
\end{lemma}

\begin{proof}
Let us first count the number of subsets of $X$ of the form $P\cap X$ where $P$ is a polyhedron with at most $k_2$ inequalities. As observed in Arora et al~\cite{arora2012computing}, this can be reduced to counting the number of subsets of the form $H \cap X$ where $H$ is a halfspace. The number of such subsets is $O(2^dm^d)$ and can be enumerated in the same amount of time (as was shown in Arora et al~\cite{arora2012computing} by a simple iterative procedure). To choose a subset of the form $P\cap X$ where $P$ is a polyhedron with at most $k_2$ inequalities, one simply needs to iteratively choose $k_2$ subsets given by halfspace intersections. Thus, there are $O((2^dm^d)^{k_2})$ such subsets and these can be enumerated in this iterative fashion. 

To finally get partitions induced by polyhedral coverings, one needs to iteratively choose $k_1$ subsets of the form $P\cap X$ where $P$ is a polyhedron with at most $k_2$ inequalities. The result follows. 
\end{proof}

Using these tools, we can now prove Theorem~\ref{thm:arora-gen}.

\begin{proof}[Proof of Theorem~\ref{thm:arora-gen}]
By Lemma~\ref{lem:partition}, there exists a partition $C_1 \uplus C_2 \uplus \ldots \uplus C_k = \{1, \ldots, m\}$, a partition $R_1 \uplus R_2 \uplus \ldots \uplus R_\ell = \{1, \ldots, n\}$, and matrices $\bar{U} \in \R^{n\times d}$, $\bar{V} \in \R^{d \times m}$ such that conditions $1, 2, 3,$ and $4$ in Lemma~\ref{lem:partition} hold.

%For each $j \in \{1,\ldots, k\}$, let $CS_j := \cup_{s\in C_j} supp(\bar{V}_s)$ where $supp(\bar{V}^s)$ denotes the support of the column $\bar{V}^s$ of $\bar{V}$. Similarly, for each $i \in \{1,\ldots, \ell\}$, let $RS_i := \cup_{t\in R_i} supp(\bar{U}_t)$ where $supp(\bar{U}_t)$ denotes the support of the row $\bar{U}_t$ of $\bar{U}$. 

%Recall that for every $j \in \{1,\ldots, k\}$, $M^s = \bar{U}\cdot\bar{V}^s$ for every $s \in C_j$. \mdnote{Have we defined $M^s$ where is an index, as opposed to $M^{\{s\}}$ where we have a collection of indices?  And why do we need this here -- next few sentences don't seem to use it?} %Similarly, for every $i \in \{1,\ldots, \ell\}$, $M_t = \bar{U}_t\cdot\bar{V}$ for every $t \in R_i$. 
Condition $3$ from Lemma~\ref{lem:partition} and Proposition~\ref{prop:invert} imply that for each $j \in \{1,\ldots, k\}$, there exist $\dim(\aff(M^{C_j})) +1 \leq d+1$ columns of $\bar V^{C_j}$, such that every other column in $\bar V^{C_j}$ can be expressed as linear combinations of these columns. Moreover, the coefficients in these linear combinations can be computed from the columns of $M^{C_j}$. Similarly, Condition $4$ from Lemma~\ref{lem:partition} and Proposition~\ref{prop:invert} imply that for every $i=1, \ldots, \ell$, the rows of $\bar U_{R_i}$ can be expressed as linear combinations of $\dim(\aff(M_{R_j})) +1 \leq d+1$ rows of $\bar U_{R_i}$.

%imply that there exist affine linear transformations $A^j : \R^d \to \R^d$, $j =1, \ldots, k$ and $B_i : \R^d \to \R^d$, $i = 1, \ldots, \ell$ such that for every $j \in \{1,\ldots, k\}$, $\bar{V}^s = A^j\cdot M^s$ for every $s \in C_j$, and for every $i \in \{1,\ldots, \ell\}$, $\bar{U}_t = B_i\cdot M_t$ for every $t \in R_i$.

Moreover, for any fixed $j \in \{1,\ldots, k\}$, since $M^{\{s\}} = \bar{U}\bar{V}^{\{s\}}$ for every $s \in C_j$, we have $\{M^{\{s\}} : s\in C_j\} \subseteq \bar{U}(P)$ since every column of $\bar V$ is in $P$ by Condition $2$ from Lemma~\ref{lem:partition}. Invoking Proposition~\ref{prop:small-desc}, we obtain that 
$\bar{U}(P)$ is described using at most $p^{2^d}$ inequalities. By Lemma~\ref{lem:partition}, $k$ is bounded by $p^d$.
%Since there are at most $2^d$ possible subsets of $\{1, \ldots, d\}$, $k \leq 2^d$. 
Therefore, $C_1, \ldots, C_k$ is a $(p^d, p^{2^d})$-polyhedral partition of $\{1, \ldots, m\}$. By Lemma~\ref{lem:poly-partition}, we can enumerate such partitions in time $O((2^dm^{d})^{p^{d+2^d}})$. Similarly, one can enumerate all possible partitions $R_1, \ldots, R_\ell$ satisfying the conditions of Lemma~\ref{lem:partition} in time $O((2^dn^{d})^{p^{d+2^d}})$.

Our algorithm will find $\bar U$ and $\bar V$ from Lemma~\ref{lem:partition}.  By conditions $1$ and $2$ of Lemma~\ref{lem:partition}, these matrices form the desired factorization of $M$.  By the discussion above, it suffices to find the partition $C_1 \uplus C_2 \uplus \ldots \uplus C_k = \{1, \ldots, m\}$, the partition $R_1 \uplus R_2 \uplus \ldots \uplus R_\ell = \{1, \ldots, n\}$ such that conditions $1, 2, 3,$ and $4$ in Lemma~\ref{lem:partition} hold, the $\dim(\aff(M^{C_j})) +1$ columns of $\bar V^{C_j}$ that form a basis for the other columns of $\bar V^{C_j}$ for each $j=1, \ldots, k$, and the $\dim(\aff(M_{R_j})) +1$ rows of $\bar U_{R_i}$ that form a basis for the other rows of $\bar U_{R_i}$ for each $i=1, \ldots, \ell$.  Once we have all of these, we can reconstruct the full $\bar U$ and $\bar V$.  

Finding these partitions and bases can be done by enumerating all possible $(p^d, p^{2^d})$-polyhedral partitions $C_1, \ldots, C_k$ of the columns of $M$, all possible $(p^d, p^{2^d})$-polyhedral partitions $R_1, \ldots, R_\ell$ of the rows of $M$, and for each choice of such partitions, introducing variables for the entries of the $\dim(\aff(M^{C_j})) +1$  special columns of $\bar V^{C_j}$, $j=1, \ldots, k$, and $\dim(\aff(M_{R_j})) +1$ special rows of $\bar U_{R_i}$, $i=1, \ldots, \ell$. Finally, set up a system of polynomial equalities in these variables that represent $M = \bar U  \bar V$. Notice that this system has only $O((k+\ell)d^2)$ variables which is a constant since $p, d$ are constants, and $k, \ell \leq p^d$. We finally invoke Theorem~\ref{thm:quantifier-elimination} to test the feasibility of such a system. Note that the requirement that $M = \bar U  \bar V$ can be expressed using $nm$ polynomial equalities, where each polynomial is a quadratic. 
\end{proof}

\section{Proof of Theorem~\ref{thm:factor}}

%We will prove the following result.
%
%\begin{theorem}\label{thm:factor}
%Let $r \in \N$ be a fixed constant. Then there exists an algorithm which, given any $\epsilon >0$ and any $n\times m$ nonnegative matrix $M$ with PSD rank $r$, can find a factorization $M=U\cdot V$ such that each row of $U$ and each column of $V$ are in $\Sym^r_+$ such that $$\| M - U\cdot V\|_\infty \leq \epsilon\| M \|_\infty$$ and has runtime polynomial in the dimensions of $M$. %({\tt Amitabh -- Should check if we can say strongly polynomial time}).
%\end{theorem}

For any square matrix $X \in \R^{r\times r}$, we use $\|X\|_{sp} := \max_{y \in \R^r\setminus \{0\}} \frac{\|Xy\|_2}{\|y\|_2}$ to denote the spectral norm of $X$. The algorithm depends on this key result (paraphrased here) from~\cite{briet2014existence}.

\begin{theorem}\label{thm:rescaling}
Let $M$ be an $n\times m$ matrix with nonnegative entries. If $M$ has a PSD factorization $M=U V$ such that the rows of $U$ and columns of $V$ are in $\Sym^r_+$, then there exists a PSD factorization $M = \bar{U} \bar{V}$ such that the rows of $\bar{U}$ and the columns of $\bar{V}$ have spectral norm bounded by $\sqrt{r\| M\|_\infty}$. %\mdnote{a PSD factorization, right?}
\end{theorem}

We outline the steps of the algorithm in Theorem~\ref{thm:factor}. Let $f(r)$ be such that for every matrix $X \in \Sym^r$, $\|X\|_\infty \leq f(r) \|X\|_{sp}$. Such an $f(r)$ must exist because all norms are equivalent on a Euclidean space, i.e., their values are the same upto a factor depending only on the dimension of the space.

\begin{enumerate}
\item Given $M$, let $\Delta = \| M \|_\infty$. Construct a polyhedral $\epsilon$-approximation of $\Sym^r_+$ with respect to the $\|\cdot\|_\infty$ norm on $\Sym^r$ -- see Theorem~\ref{thm:eps-approx} and the Remark following. Let $P$ be the polyhedron formed by the intersection of this polyhedral approximation with the cube $\{x \in \Sym^r: \|x\|_\infty \leq f(r)\sqrt{r\Delta}\}$.

\item By Theorem~\ref{thm:rescaling} and the assumption that $M$ has PSD rank $r$, we know there exists a factorization $M=\bar{U} \bar{V}$ such that the rows of $\bar{U}$ and columns of $\bar{V}$ are in the PSD cone, and their spectral norm is at most $\sqrt{r\Delta}$. Therefore, for every row $u$ of $\bar{U}$, we have $\|u\|_\infty \leq f(r)\sqrt{r\Delta}$ and similarly for the columns of $\bar{V}$. This implies that the rows of $\bar{U}$ and columns of $\bar{V}$ are in $P$. Since $\bar{U}, \bar{V}$ exist, we can employ Theorem~\ref{thm:arora-gen} to construct a factorization $M = U'V'$ such that the rows of $U'$ and the columns of $V'$ are in $P$. Note that the algorithm of Theorem~\ref{thm:arora-gen} may not produce a PSD factorization. To obtain an approximate PSD factorization, we construct matrices $U$ and $V$ by projecting each row of $U'$ to the nearest point in the PSD cone (according to the $\|\cdot\|_\infty$ norm), and similarly for the columns of $V'$. This can be done in polynomial time by invoking Proposition~\ref{prop:project}.  
\end{enumerate}

This concludes the description of the algorithm.  It remains to prove that $$\|M - U V\|_\infty \leq \epsilon\|M\|_\infty.$$
Our first step will be to use the fact that we projected from an $\epsilon$-approximation to the PSD cone.  In particular, we know that $\lVert V'^j -V^j\rVert_\infty \leq \epsilon \|V'^j\|_\infty$ for each $j \in \{1,\ldots, m\}$, and similarly $\lVert U'_i - U_i \rVert_\infty \leq \epsilon\|U'^j\|_\infty$ for each $i \in \{1,\ldots, n\}$.  This clearly implies that 
\begin{equation} \label{eq:project}
\|V' - V\|_\infty \leq \epsilon \|V'\|_\infty \text{ and } \|U' - U\|_\infty \leq \epsilon\|U'\|_\infty.
\end{equation}

Now we can analyze the approximation of our factorization:
\begin{equation}\label{eq:analysis}\begin{array}{rcl}
\|M - UV\|_\infty & = & \|U'V' - UV\|_\infty \\
& \leq & \| U'V' - U'V\|_\infty + \| U'V - UV\|_\infty \\
& \leq & r\| U'\|_\infty\|V' - V\|_\infty + r\|U' - U\|_\infty\|V\|_\infty \\
&\leq & r\| U' \|_\infty(\epsilon \|V'\|_\infty) + r\epsilon\|U'\|_\infty\|V\|_\infty
\end{array}
\end{equation}
where the first equality is from the fact that $M = U' V'$, the first inequality is from the triangle inequality, and the third is from~\eqref{eq:project}. The second inequality follows from the observation that for any matrices $A \in \R^{n\times r}, B \in \R^{r\times m}$, $\|AB\|_\infty \leq r\|A\|_\infty\|B\|_\infty$.

Since $\| V \|_\infty \leq (1+\epsilon)\|V'\|_\infty$ because of~\eqref{eq:project}, we obtain $\|M - UV\|_\infty \leq 3\epsilon r\|U'\|_\infty\|V'\|_\infty$. Since each row $u$ of $U'$ is in $P$, we have $\|u\|_\infty\leq f(r)\sqrt{r\Delta}$. Therefore, $\|U'\|_\infty \leq f(r)\sqrt{r\Delta}$. Similarly, $\|V'\|_\infty \leq f(r)\sqrt{r\Delta}$. Hence, $\|M - UV\|_\infty \leq 3\epsilon r\|U'\|_\infty\|V'\|_\infty \leq 3f(r)r^2\epsilon\Delta.$ %\mdnote{last inequality -- why?}
By redefining $\epsilon$ appropriately (in particular, letting $\epsilon'$ be the previous $\epsilon$ and letting  $\epsilon = 3 f(r) f \epsilon'$), we get that 
\begin{equation*}
\lVert M - U  V \rVert_{\infty} \leq \epsilon \lVert M \rVert_{\infty}
\end{equation*}
as desired.

\subsection{Computing on a Turing Machine} \label{sec:turing}

As mentioned in the introduction, the algorithm described above works in the real arithmetic model of computation.  However, this was only for ease of exposition.  We now show how to remove this assumption and work in the more standard Turing machine model of computation.

The assumption of real arithmetic was used in two places.  First, it was used when invoking Theorem~\ref{thm:quantifier-elimination} to solve a system of polynomial inequalities in the proof of Theorem~\ref{thm:arora-gen}.  The second time it was used was for computing the spectral decompositions in Proposition~\ref{prop:project} while projecting to the PSD cone in Step 2 above.

The first problem can be resolved by using a result of Grigor'ev and Vorobjov~\cite{grigor1988solving} which states that one can compute rational approximations to solutions of polynomial systems with integer coefficients within $\delta$ accuracy for any rational $\delta > 0$, in time that is polynomial in the parameters $\log(\frac{1}{\delta})$, maximum bit length of the coefficients, and $(sd)^{N^2}$, where $s$ is the number of inequalities, $d$ is the maximum degree, and $N$ is the number of variables (See ``Remark'' at the end of page 2 in~\cite{grigor1988solving}). This implies that one can find rational approximations for the rows and columns of $U'$ and $V'$ in Step 2 above, with the guarantee that $\|M - U'V'\|_\infty \leq O(\delta)$. Thus, in \eqref{eq:analysis}, the first line would be replaced by the inequality $\|M - UV\|_\infty \leq \|U'V' - UV\|_\infty + O(\delta)$, and this extra error term of $O(\delta)$ will carry through in all the subsequent inequalities in~\eqref{eq:analysis}. 

Further, although these rational approximations for the rows of $U'$ and the columns of $V'$ may not be in the polytope $P$ defined in Step 1 above, they will be within $O(\delta)$ distance of $P$. 

The problem of computing spectral decompositions to within any desired accuracy was shown to be possible in time polynomial in the size of the matrix and $\log(\frac{1}{\delta})$, where $\delta>0$ is the desired accuracy (under any matrix norm, and since for us the dimensions of these matrices are constants, i.e., $r\times r$, the choice of the norm also does not matter)~\cite{pan1999complexity}. This simply means that instead of projecting in to the closest point to the PSD cone, we instead project to some approximation of the closest point. However, this error can also be controlled. Note that the approximating point will also be in the PSD cone (it might just not be the closest one).

Thus, by keeping track of these additional error terms and defining the error parameters appropriately based on the given $\epsilon>0$, we can still keep the guarantee $\lVert M - U  V \rVert_{\infty} \leq \epsilon \lVert M \rVert_{\infty}$. 

\section{Open Questions}

Question~\ref{quest:1} remains the outstanding open question in the line of research on factorization algorithms with polynomial time guarantees. Another interesting direction would be generalize Theorem~\ref{thm:factor} to approximation guarantees with other norms. For example, the induced norms $\|M\|_{1,2} := \max_{x\in \R^m} \frac{\|Mx\|_2}{\|x\|_1}$ and $\|M\|_{\infty,2} := \max_{x\in \R^m} \frac{\|Mx\|_2}{\|x\|_\infty}$ were used in~\cite{gouveia2015approximate}. The authors show that approximate factorization with respect to these norms give rise to small SDP reformulations whose projections approximate a given polytope, where the geometric approximation is tightly determined by the approximation factor in the matrix factorization.

It would also be interesting to resolve the following question:

\begin{quote}
Let $r \in \N$ and $\epsilon > 0$ be fixed constants. Let $\mathcal{M}$ be the family of nonnegative matrices such that for every $M \in \mathcal{M}$, there exists another nonnegative matrix $\overline M$ such that $\lVert M - \overline M \rVert_\infty \leq \epsilon \lVert M\rVert_\infty$ and $\overline M$ admits a rank $r$ PSD factorization.

Does there exists an algorithm which, given any nonnegative matrix $M \in \mathcal{M}$, can find matrices $U$ and $V$ such that each row of $U$ and each column of $V$ are in $\Sym^r_+$ such that $$\| M - U V\|_\infty \leq O(\epsilon)\| M \|_\infty,$$ and has runtime polynomial in the dimensions of $M$?  In other words: if the input matrix $M$ is \emph{close} to a matrix with small PSD rank, can we find a low PSD-rank factorization that is a good approximation to $M$?\end{quote}

Approximate low-rank nonnnegative factorizations of matrices with high nonnegative rank have been extensively studied -- see~\cite{berry2007algorithms} for a survey of the diverse applications, and~\cite{arora2012computing} for a recent algorithm with provable guarantees on the complexity. The corresponding question for PSD factorizations is of similar interest.

\bibliographystyle{plain}
\bibliography{full-bib}

\end{document}